\documentclass[12pt]{article}
\usepackage{amsmath}
\usepackage{amssymb}
\usepackage{amsthm}
\usepackage{fullpage}
\usepackage[utf8]{inputenc}
\usepackage{mathtools}
\usepackage{url}
\usepackage[usenames,svgnames]{xcolor}
\usepackage{cite}
\usepackage{enumitem}
\usepackage{comment}
\usepackage{mathdots}
\usepackage{graphicx}
\usepackage{float}
\usepackage{enumitem}

\usepackage[pagebackref=false]{hyperref} 
\usepackage[all]{hypcap} 

\theoremstyle{plain}
\newtheorem{theorem}{Theorem}[section]
\newtheorem{lemma}[theorem]{Lemma}
\newtheorem{definition-theorem}[theorem]{Definition-Theorem}
\newtheorem{definition-proposition}[theorem]{Definition-Proposition}

\newtheorem{example}{Example}[section]
\newtheorem{examples}{Example}[subsection]

\newtheorem{remark}{Remark}[section]

\theoremstyle{definition}
\newtheorem{definition}{Definition}[section]

\numberwithin{equation}{section} 

\DeclareMathOperator{\tr}{tr}

\DeclareMathOperator{\Span}{span}

\DeclareMathOperator{\cyc}{cyc}

\DeclareMathOperator{\aut}{aut}

\def\ra{{\rightarrow}}

\def\tr{\mathrm {tr}}
\def\det{\mathrm {det}}

\def\Nor{\mathrm {Nor}}

\def\ln{\mathrm {ln}}

\def\diag{\mathrm {diag}}

\def\res{\mathop{\mathrm {res}}\limits}

\def\be{\begin{equation}}
\def\ee{\end{equation}}

\def\bea{\begin{eqnarray}}
\def\eea{\end{eqnarray}}

\def\bt{\begin{theorem}}
\def\et{\end{theorem}}

\def\bex{\begin{example}\small \rm}
\def\eex{\end{example}}

\def\bexs{\begin{examples}\small \rm}
\def\eexs{\end{examples}}

\def\ra{\rightarrow}

\def\br{\begin{remark}\small \rm}
\def\er{\end{remark}}

\def\res{\mathop{\mathrm{res}}\limits}

\def\&{&{\hskip -20pt}}

\def\CC{\mathcal{C}}
\def\DD{\mathcal{D}}

\def\WW{\mathcal{W}}

\def\Cb{\mathbf{C}}

\def\Ib{\mathbf{I}}

\def\Nb{\mathbf{N}}

\def\Nb{\mathbf{N}}
\def\Pb{\mathbf{P}}

\def\Zb{\mathbf{Z}}

\def\Rbb{\mathbb{R}}

\def\Zbb{\mathbb{Z}}

\def\Rbb{\mathbb{R}}

\def\Zbb{\mathbb{Z}}

\def \d {\mathrm d}

\begin{document}
\baselineskip 16pt
\medskip
\begin{center}
\begin{Large}\fontfamily{cmss}
\fontsize{17pt}{27pt}
\selectfont
	\textbf{Matrix model generating function for \\ quantum weighted Hurwitz numbers }
	\end{Large}
\\
\bigskip \bigskip
\begin{large}  J. Harnad$^{1, 2}$\footnote {e-mail: harnad@crm.umontreal.ca}  
and B. Runov$^{1,2}$\footnote {e-mail: runov@crm.umontreal.ca}
 \end{large}\\
\bigskip
\begin{small}
$^{1}${\em Department of Mathematics and Statistics, Concordia University\\ 1455 de Maisonneuve Blvd.~W.~Montreal, QC H3G 1M8  Canada}\\
\smallskip
$^{2}${\em Centre de recherches math\'ematiques, Universit\'e de Montr\'eal, \\C.~P.~6128, succ. centre ville, Montr\'eal, QC H3C 3J7  Canada}\\
 \smallskip
$^{3}${\em SISSA/ISAS, via Bonomea 265, Trieste, Italy }
\smallskip
\end{small}
\end{center}

\begin{small}\begin{center}
\today
\end{center}\end{small}
\medskip
\begin{abstract}
The  KP  $\tau$-function of hypergeometric type serving as generating function for quantum weighted Hurwitz numbers is
used to compute the Baker function and the corresponding adapted basis elements, expressed  as absolutely convergent Laurent series
in the spectral parameter. These  are equivalently expressed as Mellin-Barnes integrals, analogously to Meijer $G$-functions,
but with an infinite product of $\Gamma$-functions as integral kernel.  A matrix model representation is derived for the $\tau$-function 
evaluated  at trace invariants of an externally coupled matrix.
\end{abstract}

\section{Introduction}

Hurwitz numbers enumerate branched coverings of the Riemann sphere, with specified ramification profiles
at the branch points, and have been studied since the pioneering works of Hurwitz \cite{Hu1, Hu2}. 
Equivalently, they may be viewed as combinatorial invariants enumerating factorization of elements in the symmetric group $S_n$
as products of elements in specified conjugacy classes \cite{Frob1, Frob2}.

There has been considerable interest in recent years in making use of $\tau$-functions, which are dynamical generating functions
for solutions of classical integrable hierarchies, such as the KP or 2D Toda hierarchies, rather 
as generating functions in the combinatorial sense, for various types of enumerative geometrical and
topological invariants related to Riemann surfaces. 
Pandharipande and Okounkov \cite{Ok, OP}  showed that  special cases
of $\tau$-functions of  {\em hypegeometric} type \cite{OrSc1, OrSc2} may be viewed as generating functions 
for {\em simple} Hurwitz numbers (which enumerate branched coverings in which all but one, or two, of the branch
points have simple ramification profiles.)
Several other instances of  $\tau$-functions of hypergeometric type were shown to serve as as generating functions for
{\em weighted} single or double Hurwitz numbers $H^d_G(\mu)$ , $H^d_G(\mu, \nu)$  \cite{GH1, GH2, H1, HO}.
In many cases representations of such $\tau$-functions as random matrix integrals for various classes of measures
were found. These include: {\em simple} Hurwitz numbers \cite{OP, Ok, MM, BM, BEMS};
{\em weakly monotonic} Hurwitz numbers\cite{GGN1, HC}; {\em strongly monotonic} Hurwitz numbers and, more generally,
polynomially or rationally weighted Hurwitz numbers \cite{MM, AMMN, BEMS, AC1, AC2, AC3, BHR, Z, KZ}. 

Here, we derive a new matrix model representation for the KP $\tau$-function $\tau^{(H_q,\beta)}\left({\bf t}\right)$  generating 
 {\em quantum weighted} Hurwitz single numbers $H^d_{H_q}(\mu)$  \cite{GH2, H1, H2}.
 The main result (Theorem \ref{thm_matrix_integral_tau_H_q}) is that, when the flow parameters ${\bf t} = (t_1, t_2, \cdots, )$ are set equal to the trace invariants
\be
t_i = {1\over i}\tr(X^i) = {1\over i} \sum_{a-1}^nx_a^i, 
\ee
of a given $n\times n$ matrix $X$, with eigenvalues $(x_1, \dots, x_n)$, denoted $[X]$, the $\tau$-function
 $\tau^{(H_q,\beta)} (\big[ X \big])$ is expressible as the product 
\be
\tau^{(H_q,\beta)} (\big[ X \big])= {\beta^{{1\over 2}n(n-1)} (\prod_{i=1}^n x_i^{n-1})\Delta(\ln(x)) \over (\prod_{i=1}^n i!) \Delta(x) } \Zb_{d\mu_q}(\ln(X))
\ee
of a simple explicit Vandermonde determinantal  factor depending on the $x_a$'s and a generalized Br\'ezin-Hikami \cite{BrH} matrix integral
of the form
\be
\Zb_{d\mu_q}(Y) =\int_{M \in \Nor^{n\times n}_{\CC_n}}d\mu_q(M) e^{\tr Y M}.
\ee
Here
\be
d\mu_q(M) = \d\mu_0(M) \det(A_{H_q,n}(M))
\ee
is a conjugation invariant measure on the space $\Nor^{n\times n}_{\CC_n}$ of $n \times n$ normal matrices 
with eigenvalues $\zeta_i \in \CC_n$ supported on a contour $\CC_n$ in the complex plane, surrounding
simple poles of the integrand at the integers $-n, -n+1, \cdots$, $d\mu_0(M)$ is the Lebesgue
 measure on $\Nor^{n\times n}_{\CC_n}$ and $A_{H_q,n}(z))$ is a convergent infinite product of Euler
 $\Gamma$-functions, as defined in eq.(\ref{AH_qk}) of Section  \ref{mellin_integral}.

In Section \ref{gen_fns_quantum_Hurwitz}, the definition of weighted Hurwitz numbers, as
introduced in \cite{GH1,GH2,H1, HO}, is recalled and the KP $\tau$-function that serves as generating function
for these is defined, focussing on the case $\tau^{(H_q, \beta)}$ of quantum weighted generating functions. 
Following \cite{ACEH1, ACEH2, ACEH3},  we also define a natural basis $\{\phi^{(H_q, \beta)}_i\}_{i\in \Nb^+}$ for the element $W^{(H_q, \beta)}$ 
of the infinite Sato-Segal-Wilson Grassmannian  corresponding to this $\tau$-function, expressed as convergent
Laurent series, and the recursion operators relating them.

 In Section \ref{mellin_integral},  we derive a Mellin-Barnes integral representation for the $\phi^{(H_q,\beta)}_i$'s, analogous 
 to the one for the Meijer $G$-functions that appear in the case of rationally weight generating functions \cite{BHR}, but with integral kernels consisting of 
convergent infinite products of $\Gamma$-functions depending on the quantum parameter $q$.
  In Section \ref{tau_trace_invars}, the recursions  relating the $\phi^{(H_q,\beta)}_i$'s are applied to the finite determinantal expression arising
  when the KP flow parameters are equated to the trace invariants of a finite dimensional matrix, giving a Wronskian form for the
  $\tau$-function. This serves, in Section \ref{MM_tau}, as the link to expressing it as a matrix integral of 
  generalized Br\'ezin-Hikami \cite{BrH} type.  
 
\section{Generating functions for quantum Hurwitz numbers}
\label{gen_fns_quantum_Hurwitz}

\subsection{Weighted Hurwitz numbers and weight generating functions}
\label{gen_fns_weighted_Hurwitz}

We recall the definition of {\em pure} Hurwitz numbers \cite{Frob1, Frob2, Hu1, Hu2, LZ} and {\em weighted} Hurwitz numbers \cite{GH1, GH2, HO, H1}. 
\begin{definition}[Combinatorial]
For a  set of $k$ partitions $\{\mu^{i)}\}_{i=1,\dots, k}$ of $N\in \Nb^+$, the {\em pure} Hurwitz number 
$H(\mu^{(1)}, \dots, \mu^{(k)})$  is ${1\over N!}$ times the number of distinct ways that the identity element $\Ib_N \in S_N$ in the symmetric group in $N$ elements can be expressed as a product
\be
\Ib_N = h_1, \cdots h_k
\ee
of $k$ elements $\{h_i\in S_N\}_{i=1, \dots, k}$ such that for each $i$, $h_i$ belongs
to the conjugacy class $\cyc(\mu^{(i)})$ whose cycle lengths are equal to the parts of $\mu^{(i)}$:
\be
h_i \in \cyc(\mu^{(i)}), \quad i =1, \dots, k.
\ee
\end{definition}
An equivalent definition involves the enumeration of branched coverings of the Riemann sphere.
\begin{definition}[Geometric]
For a set of partitions $\{\mu^{(i)} \}_{i=1,\dots, k}$ of weight $|\mu^{(i)}|=N$, 
the pure Hurwitz number  $H(\mu^{(1)}, \dots, \mu^{(k)})$ is defined geometrically  as the number
of inequivalent $N$-fold branched coverings  $\CC \ra \Pb^1$  of the Riemann sphere with $k$ branch points $(Q^{(1)}, \dots, Q^{(k)})$, 
whose ramification profiles are given by the partitions $\{\mu^{(1)}, \dots, \mu^{(k)}\}$, 
normalized by the inverse  $1/|\aut (\CC)|$ of the order of the automorphism group of the covering. 
\end{definition}
The equivalence of the two follows from  the monodromy homomorphism from the fundamental group of 
$\Pb^1/\{Q^{(1)}, \dots, Q^{(k)}\}$, the Riemann sphere punctured at the branch points, into $S_N$,
 obtained by lifting closed loops from the base to the covering \cite{LZ}  .

 To define {\em weighted} Hurwitz numbers \cite{GH1, GH2, HO, H1}, we introduce a weight generating function $G(z)$, either as 
an infinite product
\be
G(z) =\prod_{i=1}^\infty (1 + c_i z)
\label{G_z_prod}
\ee
or an infinite sum
\be
G(z) = 1 + \sum_{i=1}^\infty g_i z^i,
\label{G_z_taylor}
\ee
either formally, or under suitable convergence conditions imposed upon the parameters.
Alternatively, it may be chosen in the dual form
\be
\tilde{G}(z) = \prod_{i=1}^\infty (1-c_i z)^{-1},
\label{G_tilde_z_prod}
\ee
which may also be developed as an infinite sum, 
\be
\tilde{G}(z)= 1+ \sum_{i=1}^\infty \tilde{g}_i z^i.
\label{G_tilde_z_taylor}
\ee

The independent parameters determining the weighting may be viewed as either the Taylor coefficients $\{g_i\}_{i\in \Nb^+}$,
$\{\tilde{g}_i\}_{i\in \Nb^+}$ or the parameters $\{c_i\}_{i\in \Nb^+}$ appearing in the infinite product formulae (\ref{G_z_prod}), (\ref{G_tilde_z_prod}).
They are related by the fact that (\ref{G_z_prod}) and (\ref{G_tilde_z_prod}) are generating functions for 
elementary and complete symmetric functions, respectively,
\be
g_i = e_i({\bf c}), \quad \tilde{g}_i = h_i({\bf c}),
\ee
in the parameters ${\bf c} =(c_1, c_2, \dots)$.

\begin{definition}[Weighted Hurwitz numbers]
For the case of weight generating functions of the form (\ref{G_z_prod}), choose a nonnegative integer $d$ and a fixed partition $\mu$
of weight  $|\mu| =N$. The weighted (single) Hurwitz number $H^d_G(\mu)$ is then defined \cite{GH2, HO} as the weighted sum over 
all $k$-tuples $(\mu^{(1)}, \dots, \mu^{(k)})$
\be
H^d_G(\mu) := 
\sum_{k=1}^d \sum_{\mu^{(1)}, \dots \mu^{(k)}, \  |\mu^{(i)}| =N  \atop \sum_{i=1}^k \ell^*(\mu^{(i)}) =d} 
\WW_G(\mu^{(1)}, \dots, \mu^{(k)}) H(\mu^{(1)}, \dots, \mu^{k)}, \mu)
\label{H_d_G_def}
\ee
where
\be
\ell^*(\mu^{(i)}) := |\mu^{(i)}| - \ell(\mu^{(i)})
\ee
is the {\em colength} of the partition $\mu^{(i)}$, and the weight factor is defined to be
\bea
 \WW_G(\mu^{(1)}, \dots, \mu^{(k)}) &\&:=
 {1\over k!}
 \sum_{\sigma \in S_{k}} 
 \sum_{1 \leq b_1 < \cdots < b_{k }} 
  c_{b_{\sigma(1)}}^{\ell^*(\mu^{(1)})} \cdots c_{b_{\sigma(k)}}^{\ell^*(\mu^{(k)})} ,
   \label{WG_def}
\eea
which, up to a normalization factor, is the {\em monomial symmetric function}  $m_\lambda({\bf c})$  \cite{Mac}
in the variables ${\bf c}= (c_1, c_2, \dots, )$ corresponding to the partitions $\lambda$ of length $k$
and weight $d =\sum_{i=1}^k \ell^*(\mu^{(i)})$ whose parts are equal to the colengths $\{ \ell^*(\mu^{(i)})\}_{i=1, \dots, k}$.

For the case of dual generating functions of the form (\ref{G_tilde_z_prod}),  the weighted 
(single) Hurwitz number $H^d_{\tilde{G}}(\mu)$ is defined  \cite{GH2, HO}  as the weighted sum 
\be
H^d_{\tilde{G}}(\mu) := 
\sum_{k=1}^d \sum_{\mu^{(1)}, \dots \mu^{(k)}, \  |\mu^{(i)}| =N  \atop \sum_{i=1}^k \ell^*(\mu^{(i)}) =d} 
\widetilde{\WW}_{\tilde{G}}(\mu^{(1)}, \dots, \mu^{(k)}) H(\mu^{(1)}, \dots, \mu^{k)}, \mu),
\label{H_d_G_tilde_def}
\ee
where the weight factor is defined as
\bea
\widetilde{\WW}_{\tilde{G}}(\mu^{(1)}, \dots, \mu^{(k)}) &\&:=
 {(-1)^{k+\sum_{i=1}^k\ell^*(\mu^{(i)})}\over k!}
 \sum_{\sigma \in S_{k}} 
 \sum_{1 \le b_1 \le \cdots \le b_k } 
  c_{b_{\sigma(1)}}^{\ell^*(\mu^{(1)})} \cdots c_{b_{\sigma(k)}}^{\ell^*(\mu^{(k)})},
   \label{WG_tilde_def}
\eea
which, again up to a normalization factor, is the {\em forgotten symmetric function} $f_\lambda({\bf c})$ \cite{Mac}
in the variables ${\bf c} = (c_1, c_2, \dots, )$ corresponding to the partition $\lambda$ of length $k$
and weight $d =\sum_{i=1}^k \ell^*(\mu^{(i)})$  with parts equal to the colengths $\{ \ell^*(\mu^{(i)})\}_{i=1, \dots, k}$.

\end{definition}

\subsection{Quantum weighted Hurwitz numbers}

In the following, we  consider a variant of the case of {\em quantum} weighted Hurwitz numbers \cite{H1, GH2}, 
with weight generating function of dual type (\ref{G_tilde_z_prod}), with the constants $c_i$  chosen to be
\be
c_i = q^{i-1}, \quad i=1, 2, \dots.
\ee
for some parameter $q$ with $0 <|q|<1$, so that
\be
\tilde{G}(z) = H_q(z) := \prod_{i=0}^\infty(1 - q^i z)^{-1} = \sum_{n=0}^\infty {z^n \over (q; q)_n},
\label{H_q_z}
\ee
where
\be
(a ; q)_n := (1-a)(1-aq) \cdots (1-aq^{n-1})
\ee
is the $q$-Pochhammer symbol. 
\begin{remark}
This may be viewed as the scaled $q$-exponential function \cite{Andrews}, since
\be
\lim_{q\ra 1} H_q((1-q)z) = e^z. 
\ee
\end{remark}

Specializing (\ref{H_d_G_tilde_def}), ({\ref{WG_tilde_def}) to this case gives:
\begin{definition}[Quantum weighted Hurwitz numbers]
For weight generating function  (\ref{H_q_z}),  choosing a nonnegative integer $d$ 
and a fixed partition $\mu$ of $N$, the quantum weighted (single) Hurwitz number $H^d_{H_q}(\mu)$ is
 the weighted sum over all $k$-tuples $(\mu^{(1)}, \dots, \mu^{(k)})$
\be
H^d_{H_q}(\mu) := 
\sum_{k=1}^d \sum_{\mu^(1), \dots \mu^{(k)}, \  |\mu^{(i)}| =N  \atop \sum_{i=1}^k \ell^*(\mu^{(i)}) =d} 
\widetilde{\WW}_{H_q}(\mu^{(1)}, \dots, \mu^{(k)}) H(\mu^{(1)}, \dots, \mu^{k)}, \mu)
\label{H_d_G_def}
\ee
where
\be
\ell^*(\mu^{(i)}) := |\mu^{(i)}| - \ell(\mu^{(i)})
\ee
is the {\em colength} of the partition $\mu^{(i)}$, and the weight factor is 
\bea
\widetilde{\WW}_{H_q}(\mu^{(1)}, \dots, \mu^{(k)})&\&:=
 {(-1)^{d-k}\over k!}
 \sum_{\sigma \in S_{k}} 
\sum_{1 \le b_1 \le  \cdots \le  b_{k } } 
  q^{(b_{\sigma(1)}-1)\ell^*(\mu^{(1)})} \cdots q^{(b_{\sigma(k)}-1)\ell^*(\mu^{(k)})}  \cr
  &\&=  {(-1)^{d-k}\over k!} \sum_{\sigma \in S_k} \prod_{j=1}^k  {1  \over (1 - q^{\sum_{i=1}^j\ell^*(\mu^{(\sigma(i))}})}.
   \label{WG_def}
\eea
\end{definition}

\begin{remark}[Riemann-Hurwitz formula]
If $d$ is defined  as in the constraint on the weighted sums (\ref{H_d_G_def}),
 \be
d = \sum_{i=1}^k \ell^*(\mu^{(i)}),
\ee
the {\em Riemann-Hurwitz formula}  gives the Euler characteristic $\chi$ of the branched cover
\hbox{$\CC\ra \Pb^1$ with $k+1$} branch points $(Q^{(1)}, \dots, Q^{(k)}, Q)$  with ramification profiles 
$(\mu^{1)}, \dots, \mu^{(k)}, \mu)$ as
 \be
\chi = N +\ell(\mu) - d.
\label{Riemann_Hurwitz}
\ee
If $\CC$ is connected and $g$ is the genus of $\CC$, we have
\be
\chi = 2 - 2g.
\ee
\end{remark}

\subsection{Hypergeometric $\tau$-functions}

We next recall the definition \cite{GH2, HO, H1} of the KP $\tau$-function of hypergeometric type \cite{OrSc1, OrSc2} 
that serves as generating function for the quantum  weighted Hurwitz numbers $H^d_{H_q}(\mu, \nu)$.
For the weight generating function $H_q(z)$ and a nonzero small parameter $\beta$,
we define two doubly infinite sequences of numbers $\{r_i^{(H_q, \beta)}, \rho_i\}_{i \in \Zb}$,
labeled by the integers
\bea
r^{(H_q, \beta)}_i &\&:= \beta H_q( i\beta), \quad i \in \Zb,  \quad \rho_0 =1,
\label{r_G_beta_i_def} \\
\rho^{(H_q, \beta)}_i &\& := \prod_{j=1}^i r^{(H_q, \beta)}_j,\quad  \rho^{(H_q, \beta)}_{-i}:=\prod_{j=0}^{i-1}( r^{(H_q, \beta)}_{-j})^{-1}, 
\quad i\in \Nb^+,
\label{rho_i_def}
\eea
related by
\be
r_i^{(H_q, \beta)} = {\rho^{(H_q, \beta)}_i \over \rho^{(H_q, \beta)}_{i-1}},
\ee
where $\beta$ is chosen  such  that the $H_q(i\beta)$ does not vanish for
any integer $i\in  \Zb$. 
For each partition $\lambda$ of $N$, we define the associated {\em content product} coefficient \cite{GH2, HO, H1}
\be
r^{(H_q, \beta)}_\lambda := \prod_{(i,j) \in \lambda} r^{(H_q, \beta)}_{j-i}.
\ee

The KP $\tau$-function of hypergeometric type associated to these parameters is then defined as the Schur function series
\be
\tau^{(H_q, \beta)}({\bf t}) 
:=\sum_{N=0}^\infty \sum_{{\lambda}, \ |\lambda|=N}  
{d_{\lambda}\over N!} r_\lambda^{(H_q, \beta)}  s_\lambda({\bf t}),
\label{tau_G_beta_schur_series}
\ee
where $d_\lambda$ is the dimension of the irreducible representation of the symmetric group $S_N$ 
and ${\bf t}=(t_1, t_2 \dots)$ are the infinite sequence of KP flow parameters,  equated to the 
infinite sequence  of normalized power sum symmetric functions $(p_1, {p_2\over 2}, \dots)$ \cite{Mac}
in some auxiliary  infinite sequence of variables $(x_1, x_2, \dots )$.

We then use the Schur character formula \cite{FH, Mac}
  \be
  s_\lambda = \sum_{\mu, \ |\mu|=|\lambda|} {\chi_\lambda(\mu) \over z_\mu} p_\mu,
  \label{schur_character_formula}
  \ee
  where $\chi_\lambda(\mu)$ is the irreducible character of the $S_N$ representation 
  determined by $\lambda$ evaluated on the conjugacy class $\cyc(\mu)$ consisting of elements with cycle lengths
  equal to the parts of $\mu$, and 
  \be
  z_{\mu} = \prod_{i=1}^{\ell(\mu)} m_i(\mu)! i^{m_i(\mu)}
  \label{mu_stabilizer_order}
  \ee 
  is the order of the stabilizer of the elements of this conjugacy class,  to re-express the Schur function
   series (\ref{tau_G_beta_schur_series}) as an expansion  in the basis $\{p_\mu\}$ 
  of power sum symmetric functions, where
\be
  p_\mu:= \prod_{i=1}^{\ell(\mu)} p_{\mu_i}, \quad p_j = j t_j, \ j\in \Nb.
 \ee
   \begin{theorem}[\cite{GH1, GH2, HO, H1, H2}]
  The $\tau$-function $\tau^{H_q, \beta}({\bf t}) $ may equivalently be expressed as
 \be
\tau^{(H_q, \beta)}({\bf t}) 
=\sum_{\mu, \, |\mu| =d}\sum_{d=0}^\infty  \beta^d H^d_{H_q} (\mu) p_\mu({\bf t}).
 \label{tau_G_beta_power_sum_series}
 \ee
and is thus a generating function for the weighted Hurwitz numbers $H^d_{H_q} (\mu) $.
    \end{theorem}
    
\subsection{The (dual) Baker function and adapted basis}

By Sato's formula \cite{Sa, SS, SW}, the dual Baker-Akhiezer function $\Psi^*_{(G, \beta)}(z, {\bf t})$ corresponding
to the KP $\tau$-function eq.~(\ref{tau_G_beta_schur_series}) is
\be
\Psi^*_{(H_q, \beta)}(z, {\bf t}) = e^{-\sum_{i=1}^\infty t_iz^i} 
{\tau^{(H_q, \beta)}({\bf t} +[z^{-1}]) \over \tau^{(H_q, \beta)}({\bf t}) },
\ee
where
\be
[z^{-1}] := \left({1\over z}, {1\over 2 z^2}, \dots, {1\over n z^n}, \dots\right).
\ee
Evaluating at ${\bf t } ={\bf 0}$ and setting
\be
x := {1\over z},
\ee
we define
\be
\phi^{(H_q, \beta)}_1(x) := \Psi^*_{H_q, \beta}(1/z, {\bf 0}) = \tau^{(H_q, \beta)}([x]). 
\label{phi1_tau}
\ee
More generally, following \cite{ACEH1, ACEH2, ACEH3}, we introduce a sequence of  functions $\{\phi_k(x)\}_{k\in \Nb}$,
defined as  contour integrals around a loop centred at  the origin (or as formal residues)
\be
\phi_k^{(H_q, \beta)}(x) = {\beta\over 2\pi i  x^{k-1}} \oint_{\zeta=0} \rho^{(H_q, \beta)}(\zeta) e^{\beta^{-1} x\zeta } {d\zeta\over \zeta^k},
\label{phi_k_fourier_rep}
\ee
where $\rho^{(H_q, \beta)}(\zeta)$ is the Fourier series
\be
\rho^{(H_q, \beta)}(\zeta) := \sum_{i\in \Zb} \rho^{(H_q, \beta)}_i \zeta^{-i-1},
\label{rho_fourier}
\ee
with the $\rho^{(H_q, \beta)}_i$'s given by eqs.~(\ref{r_G_beta_i_def}), (\ref{rho_i_def}).
Then $\{\phi^{(H_q, \beta)}_k(1/z)\}_{k\in \Nb^+}$ forms a basis for the element $W^{(H_q,\beta)}$ of the infinite 
Sato-Segal-Wilson Grassmannian corresponding to the $\tau$-function $\tau^{(H_q,\beta)}({\bf t})$.

 The $\phi^{(H_q, \beta)}_k$'s may alternatively be expanded as Laurent series by evaluating the integrals as a sum of residues at the origin, 
\be
\phi^{(H_q, \beta)}_{k}(x)  = \beta x^{1-k}\sum_{j= 0}^{\infty} 
{ \rho^{(H_q, \beta)}_{j-k}\over j!} \left({x \over\beta}\right)^j.
\label{phi_k_series}
\ee
It follows that  these satisfy the recursive sequence of equations
\be
\beta (\DD +k -1) \phi^{(H_q, \beta)}_k= \phi^{(H_q, \beta)}_{k-1}, \quad k \in \Zb,
\label{D_phi_k}
\ee

We then have the following results regarding convergence of the Taylor series (\ref{H_q_z}) and (\ref{phi_k_series}) and
the asymptotic form of $H_q(z)$ for large $|z|$,  which are all proved in Appendix \ref{app_A}.

\begin{lemma} 
\label{HqTaylorLemma}
For any $\beta<0$, $j\ge 1$ the radius of convergence of the Taylor series of the function $\log H_q(\beta(j+z))$ is greater than $\frac{1}{2}$.
\end{lemma}

\begin{lemma}
\label{conv_phi_series}
The series (\ref{phi_k_series}) is absolutely convergent for all $x$ provided $\beta <0$. 
 \end{lemma}

 \begin{lemma}
\label{HqAsLemma}
	The asymptotic form of $H_q(z)$ for large $|z|$ in the left half plane is given by
	\be
	\label{HqAs}
		\log H_q(z) \sim \frac{\log^2(-z)}{2\log q}-\frac{\log(-z)}{2\log q} -S(-z) -C_q+o(1)
	\ee
	with
	\be
	\label{SHqDef}
		S(z)=\sum\limits_{k=1}^{\infty}\frac{\cos\left(\frac{2\pi k\log z}{\log q}\right)}{\sinh\left(\frac{2\pi^2 k}{\log q}\right)}\,\quad
	C_q=-\frac{\pi^2}{6}\left(\frac{1}{\log q}+\frac{\log q}{2\pi^2}\right)\,.
	\ee	
\end{lemma}

Following \cite{ACEH1, ACEH2, ACEH3}, we introduce the recursion operator
\be
R := \beta x H_q(\beta\DD),
\ee
where $\DD$ is the Euler operator
\be
\DD := x {\d \over dx},
\ee
and verify that the $\phi^{(H_q, \beta)}_k$'s also satisfy the recursion relations
\be
R(\phi^{(H_q, \beta)}_k)= \phi^{(H_q, \beta)}_{k-1},  \quad  k \in \Zb,
\label{phi_k_rec}
\ee
and the $k=1$ value $\phi^{(H_q, \beta)}_1(x)$ coincides with  (\ref{phi1_tau}).

\section{Mellin-Barnes integral representation of $\phi_k(x)$} 
\label{mellin_integral}

As in the case of rational weight generating functions \cite{BH}, we can equivalently represent the function $\phi^{(H_q, \beta)}_k$
 in the form of a Mellin-Barnes integral, provided $\beta <0$. Define 
 \be
A_{H_q, k}(z) := (-\beta)^{1-k}\Gamma(1-k-z)\prod_{m=0}^\infty \left( (-\beta q^{m})^{-z} {\Gamma(-\beta^{-1}q^{-m}) \over \Gamma(z-\beta^{-1}q^{-m})}\right).
\label{AH_qk}
\ee
\begin{theorem}
\label{thm_mellin_rep_quantum}
The following integral representation of  $\phi^{(H_q, \beta)}_k(x)$ is valid for all $x\in \Cb$,
\be
\label{phi_irep_mellin}
\phi^{(H_q, \beta)}_k= 
\frac{1}{2\pi i}\int_{\CC_k} A_{H_q, k}(s) x^s ds, 
\ee
where the contour $\CC_k$  starts  at $+\infty$ immediately above the real axis, proceeds to the left above the axis,
winds around the poles at the integers $s= -k, -k+1. \dots$ in a counterclockwise sense and returns, just below the axis, to $\infty$.
\end{theorem}
The proof of this theorem depends on the asymptotic behaviour of $A_{H_q, k}(s) $ as $s\ra \infty$ in the right half plane, which follows
from the next two lemmas, whose proofs are given in Appendix \ref{app_A}.
\begin{lemma}
\label{AqAsLemma}
The asymptotic form of the function $A_{H_q, k}$  for large $|z|$, $\Re(z)>0$ and  $\beta < 0$, is given by
\be
	\log A_{H_q, k}(z) \sim  \frac{z\log^2(-\beta z)}{2\log q}+O(z\log z)\,,\quad \Re(z)>0.
\label{AqAs}
\ee
\end{lemma}
\begin{lemma}
\label{YAsLemma}
	The asymptotic form of the sum
	\be
		\hat{S}(z)=\sum_{k=0}^{\infty}\left[\beta^{-1}q^{-m}\log(1-\beta z q^m)+z\right],
	\ee
	for large $|z|$, $\Re(z)>0$ and $\beta<0$, is given by
	\be
		\hat{S}(z)\sim -\frac{z\log(-\beta z)}{\log q}+O(z).
	\ee
\end{lemma}

We now proceed to the proof of  Theorem~\ref{thm_mellin_rep_quantum}.
 \begin{proof}[Proof of Theorem~\ref{thm_mellin_rep_quantum}]
 The  asymptotic formula (\ref{AqAs}) ensures that the  integral in (\ref{phi_irep_mellin}) is convergent.
The poles are simple and located at the integers $\{i-k\}_{i\in \Nb}\dots $ and the residue  at $s = i-l$ is
\be
{\beta\over (i-1)!} \left({x\over \beta}\right)^{i-k} {\hskip - 10 pt}\rho^{(H_q, \beta)}_{i-k-1},
\ee
where $\rho^{(H_q, \beta)}_j$ is defined in  (\ref{rho_i_def}).
Evaluating the integral as the sum over residues at the poles  $\{s = i-k \}_{i\in \Nb}$ thus gives eq.~\ref{phi_k_series}).

\end{proof}

\section{The $\tau$-function $\tau^{(H_q,\beta)}({\bf t})$ evaluated on power sums}
\label{tau_trace_invars}

As detailed in \cite{ACEH1, ACEH2, ACEH3},  $\tau^{(H_q,\beta)}({\bf t})$ is the KP $\tau$-function  corresponding to 
the Grassmannian element $W^{(H_q, \beta)}$  spanned by the basis elements $\{\phi^{(H_q, \beta)}_i(1/z)\}_{i\in \Nb^+}$
obtained from the monomial basis $\{z^i\}_{i\in \Nb}$ by applying a suitable group element $g$, 
\be
W^{(H_q, \beta)}  = \Span\{ \phi^{(H_q, \beta)}_i(1/z) := g^{(H_q, \beta)}(z^{i-1})\},
\ee
where
\be
(g^{(H_q, \beta)}f)(z) := \res_{\zeta=0}\left( \rho^{(H_q, \beta)}(\zeta) f(z/\zeta) e^{\zeta\over \beta z}\right)
\ee
and 
\be
\rho^{(H_q, \beta)}(z) := \sum_{i\in \Zb} \rho^{(H_q, \beta)}_{-i-1} z^i.
\label{rho_fourier}
\ee
If $\tau^{(H_q,\beta)}({\bf t})$  is evaluated at the trace invariants
\be
{\bf t} = \big[ X \big], \quad t_i = {1\over i} \tr X^i
\ee
of a diagonal $n \times n$ matrix
\be
X := \diag(x_1, \dots, x_n),
\label{X_def}
\ee
it follows from the Cauchy-Binet \cite{Mac} identity that it is expressible as the ratio of $n \times n$ determinants \cite{HB}
\be
\tau^{(H_q,\beta)}\left(\big[X \big]\right) ={\prod_{i=1}^n x_i^{n-1}\over \prod_{i=1}^n \rho^{(H_q, \beta)}_{-i}} {\det\left( \phi^{(H_q, \beta)}_i(x_j)\right)_{1\le i,j, \le n} \over \Delta(x)},
\label{tau_phi_i_det}
\ee
where
\be
\Delta(x) = \prod_{1\le i < j \le n}(x_i - x_j)
\ee
is the Vandermonde determinant.

From (\ref{phi_k_series}, (\ref{D_phi_k}}), it also follows that each $\phi^{(H_q, \beta)}_k(x)$ may be expressed 
as a finite lower triangular linear combination of the powers of the Euler operator $\DD$ applied to $\phi_n(x)$:
\be
\phi^{(H_q, \beta)}_k(x) = \beta^{n-k}\DD^{n-k} \phi^{(H_q, \beta)}_n(x) + \sum_{j=0}^{n-k-1} \Gamma_{k j} \DD^j \phi^{(H_q, \beta)}_n(x), \quad  k=1, \dots, n,
\ee
where the constant coefficients $\Gamma_{kj}$ are  easily determined from (\ref{D_phi_k}).
Therefore, by elementary column operations, we have
\be
\tau^{(H_q,\beta)}\left(\big[X \big]\right) =\kappa^{(H_q, \beta)}_n  \left(\prod_{i=1}^n x_i^{n-1}\right){\det\left( \DD^{i-1}\phi^{(H_q, \beta)}_n(x_j)\right)_{1\le i,j, \le n} \over \Delta(x)},
\label{tau_phi_i_det}
\ee
where
\be
\kappa^{(H_q, \beta)}_n :={\beta^{{1\over 2}n(n-1)}\over \prod_{i=1}^n \rho^{(H_q, \beta)}_{-i}}.
\label{kappa_n}
\ee
Now define the diagonal matrix 
\be
Y = \diag (y_1, \dots y_n)
\label{Y_def}
\ee
 by
\be
X = e^Y, \quad Y=\ln(X), \quad x_i = e^{y_i}, \quad i =1 , \dots, n,
\ee
and let
\be
f^{(H_q, \beta)}_n(y) := \phi^{(H_q, \beta)}_n(e^y) = \int_{\CC_n}A_{H_q,n}(s) e^{ys} ds.
\ee
Then (\ref{tau_phi_i_det}) can be expressed as a ratio of Wronskian determinants
\be
\tau^{(H_q\beta)}\left(\big[X \big]\right) =\kappa^{(H_q, \beta)}_n  \left(\prod_{i=1}^n x_i^{n-1}\right){\det\left( (f^{(H_q, \beta)}_n)^{(i-1)}(y_j)\right)_{1\le i,j, \le n} \over \Delta(e^y)},
\label{tau_wronskian_f_det}
\ee
where
\be
\Delta(e^y):= \prod_{1\le i<j \le n}(e^{y_i} - e^{y_j}).
\ee


\section{Matrix integral  representation of $\tau^{(H_q,\beta)}\left(\big[X \big]\right)$}
\label{MM_tau}

We now state and prove our main result.

Consider the generalized Br\'ezin-Hikami \cite{BrH} matrix integral
\be
\Zb_{d\mu_q}(Y) =\int_{M \in \Nor^{n\times n}_{\CC_n}}d\mu_q(M) e^{\tr Y M}.
\label{brezin_hikami_Y}
\ee
where
\be
d\mu_q(M) = \d\mu_0(M) \det(A_{H_q,n}(M))
\label{d_mu_A_q_measure}
\ee
is a conjugation invariant measure on the space $\Nor^{n\times n}_{\CC_n}$ of $n \times n$ normal matrices (i.e., unitarily diagonalizable)
\be
M = U Z U^\dag \in \Nor^{n\times n}_{\CC_n}, \quad U \in U(n), \quad Z= \diag(\zeta_1, \dots, \zeta_n),
\ee
with eigenvalues $\zeta_i \in \CC_n$ supported on the contour $\CC_n$, and $d\mu_0(M)$ is the Lebesgue
 measure on $\Nor^{n\times n}_{\CC_n}$.

 \begin{theorem}
 \label{thm_matrix_integral_tau_H_q}
 The KP $\tau$-function $\tau^{(H_, \beta)}({\bf t})$ has the following  matrix integral representation when restricted 
 to the the trace invariants $[X]$ of an externally coupled matrix:
\be
\tau^{(H_q,\beta)} (\big[ X \big])= {\beta^{{1\over 2}n(n-1)} (\prod_{i=1}^n x_i^{n-1})\Delta(\ln(x)) \over (\prod_{i=1}^n i!) \Delta(x) } \Zb_{d\mu_q}(\ln(X)).
\label{matrix_integral_tau_H_q}
\ee
 \end{theorem}
\begin{proof}
Using the Harish-Chandra-Itzykson-Zuber integral  \cite{HC, IZ} 
\be
\int_{U \in U(n)}d\mu_{\text{H}}(U) e^{\tr(YUZ U^\dag)} = {(\prod_{i=1}^{n-1} i!) \det\left(e^{y_i \zeta_j}\right)) \over \Delta(y) \Delta(\zeta)},
\label{HCIZ_integral}
\ee
 where $d\mu_{\text{H}}(U) $ is the Haar measure on $U(n)$, to evaluate the angular integral gives
\bea
\Zb_{d\mu_q}(Y) &\&= {(\prod_{i=1}^{n-1} i!)\over \Delta(y)}\prod_{i=1}^n\left(\int_{\CC_n}  d\zeta_i A_{H_q,n}(\zeta_i) \right)\Delta(\zeta) \det\left( e^{y_i \zeta_j}\right)_{1\le i, j, \le n} \cr
&\& = { (\prod_{i=1}^n i!)\over \Delta(y)}\det\left((f^{(H_q, \beta)}_n)^{(i-1)}(y_j)\right)_{1\le i,j, \le n},
\eea
where we have used the Andr\'eiev identity \cite{An} in the second line.
By eq.~(\ref{tau_wronskian_f_det}), we therefore have the matrix integral representation (\ref{matrix_integral_tau_H_q}) 
of $\tau^{(H_q,\beta)}(\big[X\big])$.
\end{proof}

\bigskip

\bigskip

\begin{appendix}

\section{Appendix: Proofs of Lemmas \ref{conv_phi_series} - \ref{HqTaylorLemma}, \ref{AqAsLemma} and  \ref{YAsLemma} }
\label{app_A}

We here provide the proofs of the lemmas that were omitted from the body of the paper.

\begin{proof}[Proof of Lemma \ref{HqTaylorLemma} ]
	\be
		\log H_q(\beta(j+z))=\log H_q(\beta j)-\sum_{k=0}^{\infty}\log\left(1-\frac{z\beta q^k}{1-j\beta q^k}\right)
	\ee
	Since we are only interested in $|z|\le\frac{1}{2}$, the following inequality holds:
	\be
		\Big|\frac{z\beta q^k}{1-j\beta q^k}\Big|< \frac{|z|}{j}\le\frac{1}{2}\,.
	\ee
	Then each logarithm can be expanded into convergent Taylor series:
	\be
		\log H_q(\beta(j+z))=\log H_q(\beta j)-\sum_{k=0}^{\infty}\sum_{i=1}^{\infty} \frac{1}{i}\left(\frac{z\beta q^k}{1-j\beta q^k}\right)^i
	\ee
	Let us denote by $k_j$ the largest integer such that $-j\beta q^{k_j}>1$. In order to show that the order of summation can be changed without losing the convergence we split the sum as follows:
	\be
		\Big|\sum_{i=1}^{\infty} \frac{1}{i}\left(\frac{z}{j}\right)^i\sum_{k=0}^{k_j}\left(\frac{j\beta q^k}{1-j\beta q^k}\right)^i\Big|
		<\sum_{i=1}^{\infty}\frac{1}{i}\left(\frac{|z|}{j}\right)^i\frac{\log(-j\beta)}{\log q^{-1}}
	\ee
	\be
		\Big|\sum_{i=1}^{\infty}\frac{1}{i}\left(\frac{z}{j}\right)^i\sum_{k_j+1}^{\infty} \left(\frac{j\beta q^k}{1-j\beta q^k}\right)^i\Big|<
		\sum_{i=1}^{\infty} \frac{1}{i}\left(\frac{|z|}{j}\right)^i \frac{j^i|\beta|^i q^{i(k_j+1)}}{1-q^i}
		<\sum_{i=1}^{\infty}\frac{1}{i}\left(\frac{|z|}{j}\right)^i\frac{1}{1-q}
	\ee
\end{proof}

\begin{proof}[Proof of Lemma \ref{conv_phi_series} ]
	 Consider the logarithm of $\rho^{(H_q, \beta)}_m$:
	 \be
	 \label{rholog}
		\log \rho^{(H_q, \beta)}_j= j\log \beta +\sum\limits_{l=1}^{j}\log H_q(l\beta) \,.
	 \ee
	Using Lemma~\ref{HqTaylorLemma} we can
	replace the sum with an  integral as follows
	 \be
	\sum\limits_{l=1}^j\log H_q(l\beta)=
	\int\limits_{\frac{1}{2}}^{j+\frac{1}{2}}\log H_q(z \beta) d z+
	\sum_{l=1}^{j}\sum_{i=1}^{\infty} \frac{1}{4^i (2i+1)!} \frac{d^{2i}\log H_q(z\beta)}{d z^{2i}}\Big|_{z=l}\,.
	 \ee
	Derivatives of the leading term of the expansion (\ref{HqAs}) are given dy
	\be
		\frac{d^n \log^2(-z)}{dz^n}= \frac{2(n-1)!}{z^n}\left(\log(-z)-\sum_{k=1}^{n-1}\frac{1}{k}\right)\,.
	\ee
	Consequently, there exist a positive constants $M_1$ and $M_2$ such that
	\be	
		\Big| \frac{d^n \log^2(\beta(l+z))}{n! d z^n}\Big| \Big|_{z=0}<\frac{M_1\log l+M_2\log n}{l^n}\,,n\geq 2
	\ee
	The correction from higher derivatives can be estimated as
	\be
		\sum_{l=1}^{j}\sum_{i=1}^{\infty} \frac{1}{4^i (2i+1)!} \Big|\frac{d^{2i}\log H_q(z\beta)}{d z^{2i}}\Big|\Big|_{z=l}
		<\sum_{l=1}^{\infty} \frac{2(M_1+M_2)\log l}{l^2}=O(1)
	\ee
	and is therefore bounded, which results in
	\be
	\sum\limits_{l=1}^j\log H_q(l\beta)=
		\int\limits_{j_q+\frac{1}{2}}^{j+\frac{1}{2}}\log H_q(z \beta) d z+O(1)\,.
	\ee
	Substituting the asymptotic expansion (\ref{HqAs}) into the last formula one gets
	 \be
	 	\log \rho^{(H_q, \beta)}_j=\frac{j\log^2(j)}{2\log q}+O(j\log(j)).
	\ee
	 Since $\log q <0$ 
	 \be
	 	\lim\limits_{m\rightarrow \infty}\rho^{(H_q, \beta)}_m=0
	\ee
	so the series (\ref{phi_k_series}) is absolutely convergent.
 \end{proof}

\begin{proof}[Proof of Lemma \ref{HqAsLemma} ]
	The ``counting'' function  
	 \be
		\eta(\zeta)=\log(1-e^{-2\pi i \zeta})
	 \ee
	 has the following useful properties:
	\bea
	\label{zcprop}
		\res\limits_{\zeta=n} \frac{d \eta(\zeta)}{d\zeta}&=&1\,,n\in \Zbb\,,\\
		|\eta(\zeta e^{-i\omega})|&=&\exp(-2\pi \zeta\sin(\omega))+O\left(e^{-4\pi\zeta \sin(\omega)}\right)\, \omega\in (0,\pi)\,,\zeta\in \Rbb\,,\zeta\rightarrow \infty\,\\
		\eta(-\zeta^{*})&=&\left(\eta(\zeta)\right)^{*}\,.
	\eea
	It is therefore straightforward to write down the following integral representation for $\log H_q(z)$:
	 \bea
	 \label{ireprho}
	 	\log H_q(z)=-\sum\limits_{m=0}^{\infty}\log(1-zq^m)&=&
		-\int\limits_{0}^{\infty}\log(1-z q^{\zeta})d\zeta \nonumber\\-
		\frac{\omega}{\pi}\log(1-z)&-&i \sum_{\sigma=\pm 1}\sigma \int \limits_{0}^{\infty}\left[\log(1-z q^{\zeta e^{-i\sigma \omega}})\partial_{\zeta}\eta(\zeta e^{-i\sigma\omega})\right] d\zeta\,,\quad
	\eea
	 where $0<\omega<\arctan(-\frac{(\pi-|\arg(z)|)\log q}{\log(-|z|)})$. 
	Integrating by parts and moving the contour of integration for the last integral through the singularities of $\log(1-z q^{\zeta})$  to $\omega=\frac{\pi}{2}$ one gets another integral representation
	\bea
	\label{ireppi2}
	 	\sum\limits_{m=0}^{\infty}\log(1-z q^m)&=&
		\int\limits_{0}^{\infty}\log(1-z\beta q^{\zeta})d\zeta\\
	&&+\frac{1}{2}\log(1-z)-\int \limits_{0}^{\infty}\sum_{\sigma=\pm 1}
\frac{\log q}{1-z^{-1}q^{-i\sigma \zeta}}
\log(1-e^{-2\pi\zeta}) d\zeta+S(-z)\,,\nonumber
	\eea
	where the term $S(-z)$ arises as a sum of the residues at the points
	\be
		\zeta_k=\frac{i\pi(2k-1)+\log(z)}{\log q}\,.
	\ee
	The function $S(z)$ is periodic in $\log(|z|)$. The sum in (\ref{SHqDef}) is uniformly convergent for $\arg(z)\in [-\frac{\pi}{2},\frac{\pi}{2}]$
 and therefore the function $S(-z)$ is  bounded in the left half plane.
	Since $|q^{i\zeta}|=1$ the last integral can be expanded into  Taylor series in $z^{-1}$ convergent for $|z|>1$.
	 Rewriting the first integral in (\ref{ireprho}) as
\bea
	-\log q \int\limits_{0}^{\infty}\log(1-z q^{\zeta}) d\zeta&=&
	\int\limits_{1}^{\infty}\log(1+\lambda^{-1}) \lambda^{-1} d\lambda+\int\limits_{0}^{1}\log(1+\lambda)\lambda^{-1} d\lambda\\
	&&-\int\limits_{0}^{-z^{-1}} \log(1+\lambda)\lambda^{-1}d\lambda-\frac{1}{2}\int\limits_{-z^{-1}}^{1}d \log^{2}\lambda
\eea
	we can expand it in $z$ to get the statement of the lemma.
\end{proof}

\begin{proof}[Proof of Lemma~\ref{AqAsLemma}]
	Using Binet's integral representation \cite{BE} for $\Gamma(z)$ and noticing that
	\be
		-\sum_{m=0}^{\infty}\left(z-\frac{1}{2}\right)\log(1-z\beta q^m)=\left(z-\frac{1}{2}\right)\log H_q(\beta z),
	\ee
	one gets, upon cancellation of like terms, 

	\bea
	\label{binrep}
		\log A_{H_q,k}(z)=&\&
		\log\left(\Gamma(1-k-z)H_q^{z-\frac{1}{2}}(\beta z)\right)
		+\hat{S}(z)\cr
		&\&-\sum\limits_{m=0}^{\infty}\int\limits_{0}^{\infty}
		\frac{1}{t}\left(\frac{1}{2}-\frac{1}{t}+\frac{1}{e^t-1}\right)e^{t\beta^{-1} q^{-m}}\left(e^{-tz}-1\right)dt.
  \eea
The factor $\left(\frac{1}{2}-\frac{1}{t}+\frac{1}{e^t-1}\right)$ is bounded by a positive function:
\be
	0<\left(\frac{1}{2}-\frac{1}{t}+\frac{1}{e^t-1}\right)\leq \frac{t}{12}.
\ee
Therefore, provided $\Re(z)>0$, 
\be
	\Big|\int\limits_0^{\infty}\left(\frac{1}{2}-\frac{1}{t}+\frac{1}{e^t-1}\right)\frac{e^{tq^{-k}\beta^{-1}-tz}}{t} d t\Big|<
	\int\limits_0^{\infty}\Big|e^{-tq^{-k}\beta^{-1}-tz}\Big| \frac{d t}{12}=
	\frac{1}{12(\Re(z)-q^{-k}\beta^{-1})}.
\ee
Furthermore, since
\be
	\sum_{k=0}^{\infty} \frac{1}{(z-q^{-k}\beta^{-1})}= -\frac{d \log(H_q(\beta z))}{d z},
\ee
taking into account the asymptotic behaviour of $H_q$ (\ref{HqAs}) gives
\be
\label{binetintas}
	\sum_{k=1}^{\infty}\int\limits_0^{\infty}\left(\frac{1}{2}-\frac{1}{t}+\frac{1}{e^t-1}\right)\frac{e^{tq^{-k}\beta^{-1}-tz}}{t} d t=O\left(\frac{\log z}{z}\right)
\ee
Combining (\ref{binrep}),(\ref{binetintas}) and using Lemmas~\ref{HqAsLemma}  and \ref{YAsLemma} we obtain
\be
	\log A_{H_q,k}(z)=z\log(H_q(\beta z))+O(z\log z)=\frac{z\log^2(-\beta z)}{2\log q}+O(z\log z).
\ee
\end{proof}

\begin{proof}[Proof of Lemma \ref{YAsLemma}]
	We repeat the steps leading to (\ref{ireppi2}) in the proof of the Lemma~\ref{HqAsLemma} in order to get the following integral representation:
\bea
\label{YIrep}
	\hat{S}(z)&=&\int\limits_{0}^{\infty}\left[\beta^{-1}q^{-\zeta}\log(1-z\beta q^{\zeta}) +z\right]d \zeta
	+\frac{1}{2\beta}\log(1-\beta z)+\frac{ z}{2}\\
	&+&\frac{\log q}{2\pi}\sum_{\sigma=\pm 1}
  \int\limits_{0}^{\infty}
	h(i\sigma \zeta ,z)
     \log\left(1 - e^{-2\pi\zeta}\right)d\zeta+\beta^{-1}\tilde{S}(-\beta z)\,,\nonumber
\eea
where
\be
	h(\theta,z)=-\frac{1}{\beta}q^{-\theta}\log(1 -z\beta q^{\theta}) -
       \frac{z}{1 -z\beta q^{\theta}}
\ee
and
\be
\label{Ysum}
	\tilde{S}(z)=
2\pi z\sum_{k=1}^{\infty}\frac
{(\sin(
     \frac{2\pi k \log(z)}{\log q}) - 
    \frac{2\pi k}{\log q}\cos(\frac{2\pi k\log(z)}{\log q}))}
{\sinh(\frac{2\pi^2 k}{\log q})(1 + \frac{4\pi^2 k^2}{(\log q)^2})}\,.
\ee
The sum $\tilde{S}$ is bounded by a linear function.
The last integral in (\ref{YIrep}) grows like $\log(z)$.
The leading contribution  therefore comes from the first integral, which can be computed analytically:
\be
\label{YLead}
\int\limits_{0}^{\infty}\left[\beta^{-1}q^{-\zeta}\log(1-z\beta q^{\zeta}) +z\right]d \zeta=
	\frac{z}{\log q}\left(1-
(1 -\frac{1}{\beta z})\log(1 - \beta z)\right)
\ee
Substituting  (\ref{Ysum}-\ref{YLead}) into (\ref{YIrep}) yields the statement of the lemma.
\end{proof}

\end{appendix}

 \bigskip
\noindent 
\small{ {\it Acknowledgements.} 
This work was partially supported by the Natural Sciences and Engineering Research Council of Canada (NSERC) 
and the Fonds de recherche du Qu\'ebec, Nature et technologies (FRQNT).  
\bigskip
\bigskip 


\newcommand{\arxiv}[1]{\href{http://arxiv.org/abs/#1}{arXiv:{#1}}}

\bigskip
\noindent

\end{document}